\documentclass[runningheads]{llncs}
\usepackage{amssymb,amsmath}
\usepackage[usenames,dvipsnames]{color}
\usepackage[colorlinks=true,linkcolor=blue,citecolor=Mahogany]{hyperref}
\usepackage{makecell}
\usepackage{multirow}
\usepackage{eulervm}
\usepackage{graphicx}

\usepackage{mdframed}

\usepackage[usenames,svgnames,dvipsnames]{xcolor}

\usepackage{tikz}
\usetikzlibrary{arrows,positioning}

\newdimen\prevdp
\def\leftlabel#1{\noalign{\prevdp=\prevdepth
   \kern-\prevdp\nointerlineskip\vbox to0pt{\vss\hbox{\ensuremath{#1}}}\kern\prevdp}}

\usepackage{url}

\usepackage{enumerate}

\usepackage{xspace}
\newcommand{\NP}{\ensuremath{\mathsf{NP}}\xspace}
\newcommand{\NPC}{\ensuremath{\mathsf{NP}}\text{-complete}\xspace}
\newcommand{\NPH}{\ensuremath{\mathsf{NP}}\text{-hard}\xspace}
\newcommand{\PNPH}{para-\ensuremath{\mathsf{NP}\text{-hard}}\xspace}

\newcommand{\WOH}{\ensuremath{\mathsf{W[1]}}-hard\xspace}

\newcommand{\FPT}{\ensuremath{\mathsf{FPT}}\xspace}
\newcommand{\tsat}{\ensuremath{(3,\text{B}2)}-{\sc SAT}\xspace}

\newcommand{\tvdp}{{\sc Two Vertex Disjoint Paths}\xspace}

\let\oldlambda\lambda
\renewcommand{\lambda}{\ensuremath{\oldlambda}\xspace}
\let\oldalpha\alpha
\renewcommand{\alpha}{\ensuremath{\oldalpha}\xspace}
\let\oldDelta\Delta
\renewcommand{\Delta}{\ensuremath{\oldDelta}\xspace}

\newcommand{\YES}{{\sc yes}\xspace}
\newcommand{\NO}{{\sc no}\xspace}
\newcommand{\yes}{{\sc yes}\xspace}

\newcommand{\true}{\text{{\sc true}}\xspace}
\newcommand{\false}{\text{{\sc false}}\xspace}

\newcommand{\rd}{{\sc Resolve Delegation}\xspace}
\newcommand{\mmo}{{\sc Minimum Maximum Outdegree}\xspace}
\newcommand{\rfd}{{\sc Resolve Fractional Delegation}\xspace}

\newcommand{\CC}{\ensuremath{\mathcal C}\xspace}
\newcommand{\DD}{\ensuremath{\mathcal D}\xspace}
\newcommand{\EE}{\ensuremath{\mathcal E}\xspace}
\newcommand{\FF}{\ensuremath{\mathcal F}\xspace}
\newcommand{\GG}{\ensuremath{\mathcal G}\xspace}
\newcommand{\HH}{\ensuremath{\mathcal H}\xspace}

\newcommand{\PP}{\ensuremath{\mathcal P}\xspace}

\newcommand{\TT}{\ensuremath{\mathcal T}\xspace}

\newcommand{\VV}{\ensuremath{\mathcal V}\xspace}

\newcommand{\XX}{\ensuremath{\mathcal X}\xspace}

\usepackage{nicefrac}

\newtheorem{observation}{\bf Observation}

\newcommand{\eps}{\ensuremath{\varepsilon}\xspace}
\renewcommand{\epsilon}{\eps}

\newcommand{\ignore}[1]{}

\newcommand{\pr}{\ensuremath{\prime}}
\newcommand{\prr}{\ensuremath{{\prime\prime}}}

\renewcommand{\leq}{\leqslant}
\renewcommand{\geq}{\geqslant}

\renewcommand{\le}{\leqslant}

\usepackage{cleveref}

\crefname{theorem}{Theorem}{\bf Theorems}
\crefname{observation}{Observation}{\bf Observations}
\crefname{lemma}{Lemma}{\bf Lemmata}
\crefname{corollary}{Corollary}{\bf Corollaries}
\crefname{proposition}{Proposition}{\bf Propositions}
\crefname{definition}{Definition}{\bf Definitions}
\crefname{claim}{Claim}{\bf Claims}
\crefname{reductionrule}{Reduction rule}{\bf Reduction rules}

\title{On Parameterized Complexity of Liquid Democracy}

\author{Palash Dey$^{1}$ \and Arnab Maiti$^{2}$ \and Amatya Sharma$^{3}$}
\institute{Indian Institute of Technology Kharagpur\\
 \email{palash.dey@cse.iitkgp.ac.in}$^1$, \email{arnabmaiti@iitkgp.ac.in}$^2$, \email{amatya65555@iitkgp.ac.in}$^3$}
\authorrunning{On Parameterized Complexity of Liquid Democracy}

\begin{document}

\maketitle

\begin{abstract}
	In {\em liquid democracy}, each voter either votes herself or delegates her vote to some other voter. This gives rise to what is called a delegation graph. To decide the voters who eventually votes along with the subset of voters whose votes they give, we need to {\em resolve} the cycles in the delegation graph. This gives rise to the \rd problem where we need to find an acyclic sub-graph of the delegation graph such that the number of voters whose votes they give is bounded above by some integer $\lambda$. Putting a cap on the number of voters whose votes a voter gives enable the system designer restrict the power of any individual voter. The \rd problem is already known to be \NPH. In this paper we study the parameterized complexity of this problem. We show that \rd is \PNPH with respect to parameters $\lambda$, number of sink nodes and the maximum degree of the delegation graph. We also show that \rd is \WOH even with respect to the treewidth of the delegation graph. We complement our negative results by exhibiting FPT algorithms with respect to some other parameters. We finally show that a related problem, which we call \rfd, is polynomial time solvable.
\end{abstract}

\section{Introduction}
In a {\em direct democracy}, agents vote for a candidate by themselves. In {\em liquid democracy}, the voters can delegate their votes to other agents who can vote on their behalf. Suppose voter $1$ delegates her vote to voter $2$ and voters $2$ and $3$ delegate their votes to voter $4$. Then voter $4$ has a voting power equivalent to $4$ individual votes. That is delegations are transitive. This particular feature can make liquid democracy a disruptive approach to democratic voting system. This happens because such a voting system can lead to what we call a super-voter who has a lot of voting power. So now the candidates instead of trying to appease the general public can do behind the closed door dealings with the super-voters and try to win the election in an unfair manner. In order to deal with this issue, a central mechanism ensures that no super-voter has a lot of voting power. Formally we do it as follows. We create a delegation graph where the set of vertices is the set of voters and we have a directed edge from vertex $i$ to vertex $j$ if voter $i$ delegates her vote to voter $j$. We observe that delegation graph may contain cycles. Every voter is also allowed to delegate her vote to more than one other voters and let the system decide her final delegation. We use a central mechanism to find a acyclic sub-graph of the delegation graph such that no super-voter (the vertices having out-degree $0$) has a lot of voting power. We call this problem \rd.

\subsection{Related Work}
An empirical investigation of the existence and influence of super-voters was done by \cite{kling2015voting}. They showed that the super-voters can be powerful although they seem to act in a responsible manner according to their results. There have been a few theoretical work in this area by \cite{green2015direct},\cite{christoff2017binary} and \cite{kahng2018liquid}. A detailed theoretical work especially on the approximation algorithms in this setting was done by \cite{golz2018fluid}. Some other important work in Liquid democracy includes \cite{brill2018pairwise} and \cite{caragiannis2019contribution}.

\subsection{Our Contribution}

We study parameterized complexity of the \rd problem with respect to various natural parameters. In particular, we consider the number of sink vertices ($t$), maximum allowed weight $\lambda$ of any sink in the final delegation graph, maximum degree ($\Delta$), tree-width,  number of edges deleted in optimal solution ($e_{rem}$), number of non-sink vertices ($|\VV\backslash \TT|$). The number of sink vertices corresponds to the number of influential voters which is often a small number in practice. This makes the number of sink vertices an important parameter to study. Similarly, the parameter \lambda corresponds to the ``power'' of a voter. Since the input to the problem is a graph, it is natural to study parameters, for example, tree-width (by ignoring the directions of the edges) and the number of edges that one needs to delete in an optimal solution. We summarize our results in \Cref{tbl:summary}. We finally show that \rd is polynomial time solvable if we allow fractional delegations~[\Cref{thm:fractional_delegation}].


\begin{table}[h]
	\centering
	\begin{tabular}{|c|c|}\hline
		\textbf{Parameter} & \textbf{Result}\\\hline\hline
		
		t & \PNPH~[\Cref{thm:t}]\\\hline
		
		(\lambda, \Delta) & \PNPH~[\Cref{thm:lambda_D}]\\\hline
		
		(\lambda, t) & quadratic vertex kernel~[\Cref{obs:lambda_t}]\\\hline
		
		tree-width & W[1]-Hard~[\Cref{thm:treewidth}] \\\hline
		
		$e_{rem}$ & FPT by bounded search tree technique~[\Cref{edges_deleted_FPT}] \\\hline\hline
		$|\VV\backslash \TT|$ & FPT by bounded search tree technique~[\Cref{non_sink_FPT}] \\\hline\hline
		
		
		\textbf{Problem under Assumption} & \textbf{Result} \\\hline\hline

		fractional delegation & Reduction to LP ~[\Cref{thm:fractional_delegation}]\\\hline     
		DAG,Bipartite Graph & W[1]-Hard w.r.t treewidth ~[\Cref{cor:bip}]\\\hline
		DAG,Bipartite Graph & para-NP-hard w.r.t \lambda, \Delta ~[\Cref{cor:lambda_D_bipartite}]\\\hline
	\end{tabular}
	\caption{Summary of results.}\label{tbl:summary}
\end{table}

\section{Preliminaries}\label{sec:prelim}
A directed graph $\GG$ is a tuple $(\VV,\EE)$ where $\EE\subseteq\{(x,y): x,y\in\VV,x\ne y\}$. For a graph \GG, we denote its set of vertices by $\VV[\GG]$, its set of edges by $\EE[\GG]$, the number of vertices by $n$, and the number of edges by $m$. Given a graph $\GG=(\VV,\EE)$, a sub-graph $\HH=(\VV^\pr,\EE^\pr)$ is a graph such that (i) $\VV^\pr\subseteq\VV$, (ii) $\EE^\pr\subseteq\EE$, and (iii) for every $(x,y)\in\EE^\pr$, we have $x,y\in\VV^\pr$. A sub-graph \HH of a graph \GG is called a {\em spanning sub-graph} if $\VV[\HH]=\VV[\GG]$ and {\em induced sub-graph} if $\EE[\HH]=\{(x,y)\in\EE[\GG]: x,y\in\VV[\HH]\}$. Given an induced path $P$ of a graph, we define \emph{end vertex} as vertex with $0$ outdegree in $P$ and \emph{start vertex} as a vertex with $0$ indegree in $P$.

\subsection{Problem Definition}

We now define our problem formally.

\begin{mdframed}
	\begin{definition}[\rd]\label{def:rd}
		Given a directed graph $\GG=(\VV,\EE)$ (also known as delegation graph) with the set $\TT\subseteq\VV$ as its set of sink vertices and an integer \lambda, decide if there exists a spanning sub-graph $\HH\subseteq\GG$ such that
		\begin{enumerate}[(i)]
			\item The out-degree of every vertex in $\VV\setminus\TT$ is exactly $1$
			\item For every sink vertex $t\in\TT$, the number of vertices (including $t$) in \VV which has a path to $t$ in the sub-graph \HH is at most \lambda
		\end{enumerate}
		We denote an arbitrary instance of \rd by $(\GG,\lambda)$. 
	\end{definition}
\end{mdframed}

In the spanning sub-graph $\HH\subseteq\GG$ , if there is a path from $u$ to $v$ in \HH such that all the vertices on this path except $v$ has out-degree $1$, then we say that vertex $u$ {\em delegates} to vertex $v$. In any spanning sub-graph $\HH\subseteq\GG$ with the out-degree of every vertex in $\VV\setminus\TT$ is exactly $1$ (we call sub-graph \HH a feasible solution), weight of a tree rooted at the sink vertex $u$ is the number of vertices (including $u$) that have a directed path to $u$. We study parameterized complexity of \rd with respect to $t$, \lambda, and the maximum degree \Delta of the input graph as our parameters. In the optimization version of \rd, we aim to minimize \lambda.
\section{Results: Algorithmic Hardness} \label{sec:hardness}

Our first result shows that \rd is \NPC even if we have only $3$ sink vertices. For that, we exhibit reduction from the \tvdp problem. 

\begin{definition}[\tvdp]
	Given a directed graph $\GG=(\VV,\EE)$, two pairs $(s_1,t_1)$ and $(s_2,t_2)$ of vertices which are all different from each other, compute if there exists two vertex disjoint paths $\PP_1$ and $\PP_2$ where $\PP_i$ is a path from $s_i$ to $t_i$ for $i\in[2]$. We denote an arbitrary instance of it by $(\GG,s_1,t_1,s_2,t_2)$.
\end{definition}

We know that \tvdp is \NPC~\cite{fortune1980directed}.
The idea is to add paths containing large number of nodes in the instance of \rd which we are creating using the instance of \tvdp. This key idea will make both the instances equivalent.

\begin{theorem}\label{thm:t}
	The \rd problem is \NPC even if we have only $3$ sink vertices. In particular, \rd is \PNPH with respect to the parameter $t$.
\end{theorem}

\begin{proof}
	The \rd problem clearly belongs to \NP. To show its \NP-hardness, we reduce from \tvdp. Let $(\GG=(\VV,\EE),s_1,t_1,s_2,t_2)$ be an arbitrary instance of \tvdp. Let $n=|\VV|$. We consider the following instance $(\GG^\pr=(\VV^\pr,\EE^\pr),\lambda)$.
	\begin{align*}
		\VV^\pr &= \{a_v: v\in\VV\} \cup \DD_1 \cup \DD_1^\pr\cup \DD_2\cup \DD_2^\pr\cup \DD_3 \text{ where}\\
		& |\DD_1|=|\DD_2^\pr|=10n, |\DD_1^\pr|=|\DD_2|=5n,|\DD_3|=15n\\
		\EE^\pr &= \{(a_u, a_v): (u,v)\in\EE \} \cup \FF 
	\end{align*}
	
	We now describe the edges in \FF. Each $\DD_1, \DD_1^\pr, \DD_2,\DD_2^\pr$ and $\DD_3$ induces a path in $\GG^\pr$ and thus the edges in these paths are part of \FF. The end vertices of the path induced on $\DD_1$ and $\DD_2$ be respectively $d_1$ and $d_2$. The start vertices of the path induced on $\DD_1^\pr$ and $\DD_2^\pr$ be respectively $d_1^\pr$ and $d_2^\pr$. The end vertices of the path induced on $\DD_1^\pr,\DD_2^\pr$ and $\DD_3$ be $t_1^\pr,t_2^\pr$ and $t_3^\pr$ respectively. The set \FF also contains the edges in $\{(d_1,a_{s_1}), (d_2,a_{s_2}), (a_{t_1},d_1^\pr), (a_{t_2},d_2^\pr)\}$. \FF also contains edge $(a_v, t_3^\pr)$ $\forall v \in \VV$. This finishes the description of \FF and thus the description of $\GG^\pr$. We observe that $\GG^\pr$ has exactly $3$ sink vertices, namely $t_1^\pr,t_2^\pr$ and $t_3^\pr$. Finally we define $\lambda=17n$. We claim that the two instances are equivalent.
	
	In one direction, let us assume that the \tvdp instance is a \YES instance. For all $i\in[2]$, let $\PP_i$ be a path from $s_i$ to $t_i$ in \GG such that $\PP_1$ and $\PP_2$ are vertex disjoint. We  build the solution \HH for \rd by first adding the set of edges $\{(u,v)|\text{outdegree of $u$ is 1}\}$. Then we add the paths $\PP_1$ and $\PP_2$. Then we add the edges $(a_{t_1},d_1^\pr), (a_{t_2},d_2^\pr)$. Then for each vertex $u$ in the set $\VV^r$ $=\{a_v | v \in \VV\} \backslash \VV[\PP_1 \cup \PP_2] $, add the edge $(u,t_3^\pr)$ to \HH.

	We observe that the out degree of every vertex is exactly $1$ in $\HH$ except the sink vertices in $\GG^\pr$ (which are $t_1^\pr,t_2^\pr$ and $t_3^\pr$). Also since \HH contains the path $\PP_i$, every vertex in $\DD_i$ has a path to $t_i^\pr$ for $i\in[2]$. Of course, every vertex in $\DD_i^\pr$ has a path to $t_i^\pr$ for $i\in[2]$ and every vertex in $\DD_3$ delegates to $t_3^\pr$. Hence $\forall i\in[3]$, the number of vertices which has a path to $t_i$ in $\HH^\pr$ is at most $16n$ which is less than \lambda. Hence the \rd instance is a \YES instance.
	
	In the other direction, let us assume that the \rd instance is a \YES instance. Let $\HH^\pr=(\VV^\pr,\EE^\prr)\subseteq\GG^\pr$ be a spanning sub-graph of $\GG^\pr$ such that (i) the out degree of every vertex which is not a sink is exactly $1$, (ii) there are at most $\lambda\;(=17n)$ vertices (including the sink nodes) in $\HH^\pr$ which has a path to $t_i^\pr$ for $i\in[3]$. Note that $a_{s_1}$ must have a path $\PP_1^\pr$ to $a_{t_1}$ in $\HH^\pr$ otherwise at least $20n$ vertices have path to either $t_{2}^\pr$ or $t_3^\pr$ in $\HH^\pr$ which is a contradiction (since $\lambda=17n$). Similarly $a_{s_2}$ must have a path $\PP_2^\pr$ to $a_{t_2}$ in $\HH^\pr$ otherwise at least $20n$ vertices have path to either $t_{1}^\pr$ or $t_3^\pr$ in $\HH^\pr$ which is a contradiction (since $\lambda=17n$). Since, for $i\in[2]$, we have a path $\PP_i^\pr$ from $a_{s_i}$ to $a_{t_i}$ in $\HH^\pr$ and the out-degree of every vertex in $\HH^\pr$ except $t_1^\pr,t_2^\pr$ and $t_3^\pr$ is $1$, the paths $\PP_1^\pr$ and $\PP_2^\pr$ are vertex disjoint. We define path $\PP_i=\{(u,v): (a_u,a_v)\in\PP_i^\pr\}$ in \GG for $i\in[2]$. Since $\PP_1^\pr$ and $\PP_2^\pr$ are vertex disjoint, it follows that $\PP_1$ and $\PP_2$ are also vertex disjoint. Thus the \rd instance is a \YES instance.

\end{proof}
\qed

We next show that \rd is \NPC even if we have $\lambda=3$ and $\Delta=3$. For that we exhibit a reduction from \tsat which is known to be \NPC~\cite{berman2004approximation}.

\begin{definition}[\tsat]
	Given a set $\XX=\{x_i: i\in[n]\}$ of $n$ variables and a set $\CC=\{C_j: j\in[m]\}$ of $m$ $3$-CNF clauses on \XX such that, for every $i\in[n]$, $x_i$ and $\bar{x}_i$ each appear in exactly $2$ clauses, compute if there exists any Boolean assignment to the variables which satisfy all the $m$ clauses simultaneously. We denote an arbitrary instance of \tsat by $(\XX,\CC)$.
\end{definition}

For each literal and clause in \tsat we add a node in the instance of \rd and we add some special set of edges and nodes so that $\lambda=3$ and both the out-degree and in-degree of every vertex is at most $3$
\begin{theorem}\label{thm:lambda_D}
The \rd problem is \NPC even if we have $\lambda=3$ and both the out-degree and in-degree of every vertex is at most $3$. In particular, \rd is \PNPH with respect to the parameter $(\lambda,\Delta)$.
\end{theorem}

\begin{proof}
	The \rd problem clearly belongs to \NP. To show its \NP-hardness, we reduce from \tsat. Let $(\XX=\{x_i:{i\in[n]}\}, \CC=\{C_j: j\in[m]\})$ be an arbitrary instance of \tsat. We define a function $f:\{x_i,\bar{x}_i: i\in[n]\}\longrightarrow\{a_i,\bar{a}_i: i\in[n]\}$ as $f(x_i)=a_i$ and $f(\bar{x}_i)=\bar{a}_i$ for $i\in[n]$. We consider the following instance $(\GG=(\VV,\EE),\lambda)$.
	\begin{align*}
		\VV &= \{a_i, \bar{a}_i, d_{i,1}, d_{i,2}: i\in[n]\} \cup \{y_j: j\in[m]\}\\
		\EE &= \{(y_j,f(l_1^j)), (y_j,f(l_2^j)), (y_j,f(l_3^j)): C_j=(l_1^j\vee l_2^j\vee l_3^j), j\in[m]\}\\
		&\cup \{(d_{i,2}, d_{i,1}), (d_{i,1}, a_i), (d_{i,1},\bar{a}_i): i\in[n]\}\\
		\lambda &= 3
	\end{align*}
	We observe that both the in-degree and out-degree of every vertex in \GG is at most $3$. Also $\Delta=3$. We now claim that the two instances are equivalent.
	
	Suppose the \tsat instance is a \YES instance. Let $g:\{x_i: i\in[n]\}\longrightarrow\{\true, \false\}$ be a satisfying assignment of the \tsat instance. We define another function $h(g,j)=f(l), j\in[m],$ for some literal $l$ which appears in the clause $C_j$ and $g$ sets it to \true. We consider the following sub-graph $\HH\subseteq\GG$
	\begin{align*}
		\EE[\HH] &= \{(d_{i,2}, d_{i,1}): i\in[n]\}\\
		&\cup \{(d_{i,1},a_i): i\in[n], g(x_i)=\false\}\\
		&\cup \{(d_{i,1},\bar{a}_i ): i\in[n], g(x_i)=\true\}\\
		&\cup \{(y_j, h(g,j)): j\in[m]\}
	\end{align*}
	We observe that \HH is a spanning sub-graph of \GG such that (i) every non-sink vertices in \GG has exactly one outgoing edge in \HH and (ii) for each sink vertex in \GG, there are at most $3$ vertices (including the sink itself) which has a path to it. Hence the \rd instance is a \YES instance.
	
	In the other direction, let the \rd instance is a \yes instance. Let $\HH\subseteq\GG$ be a sub-graph of \GG such that (i) every non-sink vertices in \GG has exactly one outgoing edge in \HH and (ii) for each sink vertex in \GG, there are at most $3$ vertices (including the sink itself) which has a path to it. We define an assignment $g:\{x_i: i\in[n]\}\longrightarrow\{\true, \false\}$ as $g(x_i)=\false$ if $(d_{i,1},a_i)\in\EE[\HH]$ and \true otherwise. We claim that $g$ is a satisfying assignment for the \tsat instance. Suppose not, then there exists a clause $C_j=(l_1^j\vee l_2^j\vee l_3^j)$ for some $j\in[m]$ whom $g$ does not satisfy. 
We define functions $f_1,f_2:\{x_i,\bar{x}_i: i\in[n]\}\longrightarrow\{d_{i,1},d_{i,2}: i\in[n]\}$ as $f_1(x_i)=f_1(\bar{x}_i)=d_{i,1}$ and $f_2(x_i)=f_2(\bar{x}_i)=d_{i,2}$. We observe that the sink vertex $f(l_i^j)$ is reachable from both $f_1(l_i^j)$ and $f_2(l_i^j)$ in \HH for every $i\in[3]$. Since $\lambda=3$, we do not have a path from $y_j$ to any of $f(l_i), i\in[3]$ which is a contradiction since the non-sink vertex $y_j$ must have out-degree $1$ in \HH.
 Hence $g$ is a satisfying assignment for the \tsat instance and thus the instance is a \yes instance.
\end{proof}
\qed
\begin{corollary}\label{cor:lambda_D_bipartite}
Given that the input graph is both bipartite and directed acyclic graph, the \rd problem is \NPC even if we have $\lambda=3$ and both the out-degree and in-degree of every vertex is at most $3$ which concludes that \rd is \PNPH with respect to the parameter $(\lambda,\Delta)$.
\end{corollary}
\begin{proof}
The corollary follows as the resulting graph \GG from reduction of \tsat instance in Theorem \ref{thm:lambda_D} is bipartite as $\VV$ can be partitioned into 2 independent sets $\VV_1=\{y_j:j\in[m]\}\cup\{d_{i,1}:i\in[n]\}$ and $\VV_2=\{a_i,\bar{a}_i,d_{i,2}:i\in[n]\}.$ Also \GG is Directed Acyclic graph as it doesn't have directed cycles.
\end{proof}
\qed

\begin{definition}
A (positive integral) $edge$ $weighting$ of a graph $G$ is a mapping $w$ that assigns to each edge of $G$ a positive integer.
\end{definition} 
\begin{definition}
An $orientation$ of $G$ is a mapping $\Lambda:E(G)\rightarrow V(G)\times V(G)$ with $\Lambda((u,v))\in \{(u,v),(v,u)\}$. 
\end{definition}
\begin{definition}
The $weighted$ $outdegree$ of a vertex $v\in V(G)$ w.r.t an edge weighting $w$ and an orientation $\Lambda$ is defined as $d^{+}_{G,w,\Lambda}(v)=\sum_{(v,u)\in E(G)\text{ with } \Lambda((v,u))=(v,u)}w((v,u))$.
\end{definition}

\begin{definition}
(\mmo).  Given a graph $G$, an edge weighting $w$ of $G$ in unary and a positive integer $r$, is there an orientation $\Lambda$ of $G$ such that $d^{+}_{G,w,\Lambda}(v)\leq r$ for each $v\in V(G)$?
\end{definition}

\begin{lemma}\cite{szeider2011not}
\mmo is W[1]-hard when parameterized by the treewidth of the instance graph
\end{lemma}

We now show that \rd is W[1]-hard when parameterized by the treewidth of the instance graph. We reduce from \mmo with instance graph $G$ to \rd  by first creating a replica of the $G$ and then taking an edge $(u,v)$ with weight $w$ and replacing it with a path of $w$ nodes with the end vertex having edges to $u$ and $v$. 

\begin{theorem}\label{thm:treewidth}
\rd is W[1]-hard when parameterized by the treewidth of the instance graph
\end{theorem}
\begin{proof}
To prove W[1]-Hardness we reduce from \mmo to \rd.  Let a graph $G(V,E)$ with an edge weighting $w$ in unary and a positive integer $r$ be an arbitrary instance of \mmo. \mmo is considered to be a YES instance if the weighted outdegree of every vertex is upper bounded by $r$.  Now using the instance of \mmo we create an instance $(\mathcal{H},r+1)$ of \rd. Let us construct a graph $\mathcal{H=(V,E)}$ where $\mathcal{V}=V_1\cup V_2$. $V_1=\{b_u:u\in V\}$. $\forall (u,v)\in E$ add the set of vertices $\{a_{uv_1},a_{uv_2},\ldots,a_{uv_{w(u,v)}}\}$ to $V_2$. $\forall (u,v) \in E$, $(a_{uv_1},b_u)\in\mathcal{E}$, $(a_{uv_1},b_v)\in\mathcal{E}$ and $\forall i \in [w(u,v)]\setminus \{1\}$, $(a_{uv_i},a_{uv_{i-1}})\in\mathcal{E}$. This completes the construction of $\mathcal{H}$ with $V_1$ as the sink nodes. It is trivial to observe the fact that $tw(\mathcal{H})\leq tw(G)+2$. We now prove that the \mmo is an YES instance iff the \rd is an YES instance\\
\\
Let \mmo be a YES instance. Let $\Lambda$ be the orientation  of $G$ which makes \mmo an YES instance. We consider the following sub-graph $\HH'\subseteq\HH$
	\begin{align*}
		\EE[\HH'] &= \{(a_{uv_i},a_{uv_{i-1}}): i\in[w(u,v)]\setminus\{1\},(u,v)\in E\}\\
		&\cup \{(a_{uv_1},b_u): (u,v)\in E, \Lambda((u,v))=(u,v)\}
	\end{align*}
	We observe that \HH' is a spanning sub-graph of \HH such that (i) every non-sink vertices in \HH has exactly one outgoing edge in \HH' and (ii) for each sink vertex in \HH, there are at most $r+1$ vertices (including the sink itself) which has a path to it. Hence the \rd instance is a \YES instance.\\
\\
Let \rd be a YES instance. Let \HH' be the spanning sub-graph of \HH which make \rd a YES instance. Let the edges in \HH' be denoted by \EE'. We consider the following orientation $\Lambda$ of $G$\\
$ \Lambda((u,v))= \left\{ \begin{array}{rcl}
(u,v) & \mbox{if} & (a_{uv_1},b_u)\in\EE'  \\ 
(v,u) & \mbox{otherwise} 
\end{array}\right.$
\\
Clearly weighted outdegree of every vertex in $G$ is atmost $r$. Therefore \mmo is an YES instance.\\
This concludes the proof of this theorem
\end{proof}
\qed
\begin{corollary}\label{cor:bip}
\rd is W[1]-hard when parameterized by the treewidth even when the input graph is both Bipartite and Directed Acyclic Graph. 
\end{corollary}
\begin{proof}
In the instance of \rd created in Theorem \ref{thm:treewidth}, graph $\mathcal{H}$ is Bipartite as there is no odd cycle in the underlying undirected graph.  Also graph $\mathcal{H}$ is Directed Acyclic Graph (DAG) as there is no directed cycle.
\end{proof}
\qed

\section{\FPT Algorithms} \label{sec:fptalg}

We now prresent our \FPT algorithms.

\begin{observation}\label{obs:lambda_t}
	There is a kernel for \rd consisting of at most $\lambda t$ vertices. In particular, there is an \FPT algorithm for the \rd problem parameterized by $(\lambda,t)$.
\end{observation}
\begin{proof}
	If the number $n$ of vertices in the input graph is more than $\lambda t$, then the instance is clearly a \NO instance. Hence, we have $n\le\lambda t$.
\end{proof}
\qed
In this section we define the notion of weights for the nodes in the subgraph \HH of the delegation graph \GG. 
 We define weight of all nodes $u$ in \GG to be $1$. To get a notion of weight of a vertex $u$ in a subgraph \HH, it can be considered as a number which is one more than the number of nodes who have delegated their vote to $u$ and then have been removed from the graph \GG during the construction of \HH. If \HH is a forest such that every non-sink node has an outdegree $1$, then clearly the weight of the tree rooted at a sink node say $t$ is sum of the weights of the nodes in the tree.
We now show \rd is FPT w.r.t number of non-sink nodes by using the technique of bounded search tree by the branching on set of vertices satisfying some key properties.
\begin{theorem}\label{non_sink_FPT}
The \rd problem has a FPT with respect to the parameter $k$ which is the number of non-sink nodes in \GG (delegation graph).
\end{theorem}
\begin{proof} Let us denote the problem instance by $(\GG,\lambda,k)$. Now we present the following reduction and branching rules.\\
 \textbf{Reduction RD.1.} If there is a vertex $v$ in $\VV$ with only one outgoing edge to a vertex $u$ ($u,v$ are distinct), delete  $v$ from graph and increase weight of $u$ by the weight of $v$. The incoming edges which were incident on $v$ (except the self loops if any) are now incident on $u$.\\
Safeness of Reduction RD.1. is trivial as a node $v$ with single outgoing edge can only delegate the votes it has got (this includes $v$'s own vote and the votes of other nodes who have delegated to $v$ so far) to the only neighbor $u$ it has got.\\
\textbf{Reduction RD.2} Remove self loops if any. \\
Safeness of Reduction RD.2. follows from the fact that no non-sink node can delegate to itself\\
\textbf{Reduction RD.3.} If \GG contains a non-sink node $v$ with outdegree more than $2(k-1)$ and indegree 0, delete $v$ from \GG. The new instance is $(\GG-v,\lambda,k-1)$\\
Safeness of Reduction RD.3. is due to the fact that if we have a vertex $v$ with outdegree greater than $2(k-1)$, it implies that it has an outgoing edge to at least $k$ sink nodes. Let us denote these sink nodes by set $S$. So, irrespective of the delegations made by other vertices, there will exist one sink node $t'\in S$ such that none of the other $k-1$ non-sink nodes have delegated to $t'$ and hence we can delegate $v$ to $t'$ and still not increase the maximum weight of the sink node.\\
\textbf{Branching B.1.}
Pick a vertex $v$ such that the outdegree is more than $2(k-1)$ and indegree is $k'>0$. Note that $k'\leq k-1$. Each of $k'$ nodes having an outgoing edge to $v$ can either delegate to $v$ or not delegate it. So we have $2^{k'}$ possibilities and hence we can create $2^{k'}$ subproblems. In each possibility if a node $u_1$ is delegating to $v$ then we delete all the outgoing edges of $u_1$ expect $(u_1,v)$ and if we have a node $u_2$ which doesn't delegate to $v$ then we delete the outgoing edge from $u_2$ to $v$.  In each of the $2^{k'}$ instances of graph created first apply R.D.1, then R.D.2, and then finally R.D.3. Now solve the problem recursively for each of the $2^{k'}$ instances created by considering each of them as a subproblem. If a non-sink node $u$ has delegated to $v$ then $u$ gets deleted due to R.D.1 and if none of the non-sink nodes delegate to $v$ then $v$ gets deleted to R.D.3. So therefore, the new parameter (number of non-sink nodes) for the smaller subproblems gets reduced by at least 1. \\
Given a directed delegation graph $\GG$, the algorithm works as follows. It first applies Reductions RD.1., RD.2.,RD.3. and Branching Rule B.1 exhaustively and in the same order. 
The parameter (number of non-sink nodes) decreases by at least $1$ for each of the subproblems as explained earlier. If we can't apply the branching rule B.1 to a given subproblem it implies that there is no non-sink node such that the outdegree is more than $2(k-1)$ and indegree is greater than 0. Also due to R.D.3 we don't have any non-sink node with outdegree more than $2(k-1)$ and indegree equal to $0$. So we can do a brute force by considering every possible delegations and solve this instance in $O(k^{k}\cdot n^{O(1)})$ running time. Note that our algorithm will only look at the feasible solutions of \rd while brute forcing for a subproblem.\\
Also since every node of bounded search tree splits into at most $2^{k-1}$ subproblems and height of the tree is $O(k)$, we get $f(k)$ leaves (where $f(k)$ is a function of $k$ only). Clearly the time taken at every node is bounded by $g(k)\cdot n^{O(1)}$ where $g(k)$ is a function of $k$ only. Thus, the total time used by the algorithm is at-most $O(f(k)\cdot g(k)\cdot n^{O(1)})$ which gives us an FPT for \rd.
\end{proof}
\qed
We now show \rd is FPT w.r.t number of edges to be deleted from delegation graph by using the technique of bounded search tree by the branching on set of edges satisfying some key properties.

\begin{theorem}\label{edges_deleted_FPT}
The \rd problem has a FPT with respect to the parameter $k$ which is the number of edges to be deleted from delegation graph.
\end{theorem}
\begin{proof}
The parameter $k$ is the number of edges to be deleted. Given any instance $\GG$ of problem , every feasible solution graph $\GG_{\TT}$ is a forest with trees  with set of roots as set of all sink nodes $\TT$. Clearly then $k=|\EE|-|\VV|+|\TT|$. Let us denote the problem instance by $(\GG,\lambda,k)$.\\
\begin{observation}\label{2deg}
If $k>0$ and only the sink nodes have outdegree 0, then there is a non-sink node with outdegree atleast $2$. 
\end{observation}
\begin{proof}
Sum of outdegree of all the non-sink nodes is greater than $|\VV|-|\TT|$. Hence the observation follows from pigeon hole principle.
\end{proof}
\textbf{Branching B.1.} Let $k>0$. Consider the vertex with maximum outdegree. If $l$ is the outdegree of one such vertex $v$, delete one of the two groups of edges $\{1,\ldots,\lfloor l/2 \rfloor\}$ and $\{\lfloor l/2 \rfloor+1, \ldots , l\}$ outgoing from $v$. Then solve the problem recursively for two new subproblems with new parameter $k^{'}\leq k-1$.

Now we describe why the Branching B.1 is safe. Note that the Branching B.1 is triggered only when $k>0$. It follows from Observation \ref{2deg} that outdegree of $v$ is at least 2. Consider the degree of $v$ to be $l$ and the corresponding outgoing edges from $v$ to be $\{1,\ldots,l\}$ . Since $v$ can delegate only to exactly one of its neighbours connected by $\{1,\ldots,l\}$, other $l-1$ edges need to be deleted from delegation graph as they can not be a part of feasible solution. If we partition the set of edges into two disjoint sets $\{1,\ldots,\lfloor l/2 \rfloor\}$ and $\{\lfloor l/2 \rfloor+1, \ldots , l\}$ , only one out of the two groups can be a part of feasible solution. This allows us to delete the other half set say $\{\lfloor l/2 \rfloor+1, \ldots , l\}$. As we know that $|l|\geq 2$ which comes from the fact that outdegree of vertex $v$ is at least 2. The problem now reduces to a smaller instance $\GG^{'}$ with edges  $\EE^{'}[\GG^\pr]$ $=$ $\EE[\GG] \backslash \{\lfloor l/2 \rfloor+1, \ldots , l\}$ and parameter number of edges to be deleted as $k^{'} \leq k-1$. Thus way we get a bounded search tree with only constant number of subproblems at each branch  such that at each recursive step the height of search tree reduces by at least one.

Given a directed delegation graph $\GG$, the algorithm works as follows. As long as $k>0$, Branching Rule B.1 is applied exhaustively in the bounded search tree.  Note that Branching Rule B.1 brings down the parameter $k$ in every call by at least 1. Whenever the parameter $k$ becomes $0$, we have a feasible solution as the non-sink nodes have the outdegree of $1$. Now we can easily check in polynomial time whether the feasible solution is a YES instance or a NO instance.
At every recursive call we decrease the parameter by at least 1 and thus the height of the tree is at most $k$. Also since every node of bounded search tree splits into two, we get $O(2^{k})$ leaves. Clearly the time taken at every node is bounded by $n^{O(1)}$. Thus if $f(k) = O(2^{k})$ be the number of nodes in the bounded search tree, the total time used by the algorithm is at most $O(2^{k}n^{O(1)})$ which gives us an FPT for \rd.
\end{proof}
\qed

\section{Structural Results}
\begin{theorem}\label{thm:fractional_delegation}
There exists a linear programming formulation for the optimization version of \rd where fractional delegation of votes is allowed. Thus the fractional variant is solvable in polynomial time.
\end{theorem}
\begin{proof}
We consider the fractional variant of Liquid Democracy Delegation Problem where it is allowed to fractionally delegate votes of a source (delegator) to multiple nodes such that total number of votes being delegated is conserved at the delegator. We formally define conservation while formulating the LP for the problem.\\
LP formulation follows similar to the LP formulation of flow-problems (e.g. Max-FLow-MinCut etc). We assign $x_{u,v}$ as weight to every edge $(u,v)\in\EE[\GG]$ which corresponds to the fractional weight of votes delegated from vertex $u$ to $v$ ( for all $u,v \in \VV[\GG]$. For all other $x_{u,v}$ where $(u,v)$ pair doesn't correspond to an edge of delegation graph we assign value 0.  It immediately follows that for all sink nodes $t\in\TT[\GG]$ , total weight of fractional votes being delegated to each sink-node $t$  (including that of the sink node $t$) is  $\sum\limits_{v\in\VV \backslash  \TT} x_{v,t}+1$                 $\forall t \in \TT$ . For all other non-sink nodes $s \in \VV \backslash \TT$ , node $s$ obeys conservation as follows :
\begin{center}
 $\sum\limits_{u\in\VV \backslash  \TT} x_{u,s}$ + 1 =  $\sum\limits_{v\in\VV} x_{s,v}$           , $\forall s \in \VV \backslash \TT$
\end{center}
Our aim is to minimize the maximum weight of votes delegated to any sink node (including that of the sink node). The corresponding LP formulation is:
\begin{center}
 $\text{minimize }z$\\
$z\geq\sum\limits_{v\in\VV \backslash  \TT} x_{v,t}+1$, $\forall t\in\TT$\\
 $\sum\limits_{u\in\VV \backslash  \TT} x_{u,s}$ + 1 =  $\sum\limits_{v\in\VV} x_{s,v}$           , $\forall s \in \VV \backslash \TT$ [Follows from conservation]\\
$x_{u,v} \geq 0$ ,$\forall (u,v) \in \EE[\GG]	$\\
$x_{u,v} = 0$ ,$\forall (u,v) \notin \EE[\GG]	$
\end{center}
\end{proof}
\qed


\section{Conclusion and Future Direction}

We have studied the parameterized complexity of a fundamental problem in liquid democracy, namely \rd. We considered various natural parameters for the problem including the number of sink vertices, maximum allowed weight of any sink in the final delegation graph, maximum degree of any vertex, tree-width,  the number of edges that one deletes in an optimal solution, number of non-sink vertices. We also show that a related problem which we call \rfd is polynomial time solvable.

An important future work is to resolve the complexity of \rd if the input graph is already acyclic or tree. We know that there exists a $\Omega(\log n)$ lower bound on the approximation factor of optimizing the maximum allowed weight of any sink~\cite{golz2018fluid}. It would be interesting to see if there exsits \FPT algorithms achieving a approximation factor of $o(\log n)$.

\bibliography{references}

\bibliographystyle{splncs04}
\newpage
\end{document}